\documentclass[11pt,a4paper,twoside,notitlepage]{article}

\usepackage{graphicx}
\usepackage[sc]{mathpazo} 
\usepackage{url}
\usepackage{array}
\usepackage{float}
\usepackage{algorithm}
\usepackage{algorithmic}
\usepackage{natbib}
\usepackage[T1]{fontenc}
\usepackage[latin1]{inputenc}
\usepackage[figuresright]{rotating}
\usepackage{amsmath,amssymb,amsfonts,xspace}
\newtheorem{prethm}{{\bf Theorem}}

\newtheorem{mydef}{\bf Definition}

\newtheorem{example}{\bf Example}

\newtheorem{prelem}{{\bf Lemma}}

\newenvironment{lem}{\begin{prelem}{\hspace{-0.5
               em}}}{\end{prelem}}

\newtheorem{preproof}{{\bf Proof}}

\newenvironment{proof}[1]{\begin{preproof}{\rm
               #1}\hfill{$\Box$}}{\end{preproof}}

\renewcommand{\thefootnote}

\title{A Computational Model and Convergence Theorem for Rumor Dissemination in Social Networks} 
\author{Masoud Amoozgar $^a$, Rasoul Ramezanian $^{b}$\\
{\footnotesize {\em $^{a,b}$Department of Mathematical Sciences,
Sharif University of Technology, Tehran, Iran}}\\
{\footnotesize {\em Laboratory of complex and multi-agent systems}}\\
{\footnotesize\texttt{{amoozgar@cs.sharif.edu}, {ramezanian@sharif.edu}}}}

\begin{document}
\maketitle
\bibliographystyle{jtbnew}

\setcitestyle{authoryear,round,aysep={}}

\begin{abstract}

The spread of rumors, which are known as unverified statements of uncertain origin, may threaten the society and its controlling is important for national security councils of countries. If it would be possible to identify factors affecting spreading a rumor (such as agents' desires, trust network, etc.), then this could be used to slow down or stop its spreading. A computational model that includes rumor features and the way a rumor is spread among society's members, based on their desires, is therefore needed. Our research is focused on the relation between the homogeneity of the society and rumor convergence in it and result shows that the homogeneity of the society is a necessary condition for convergence of the spread rumor.


\end{abstract}
\section{Introduction}
The main characteristic of a rumor is its capacity to spread uncontrollably, with an incredible rapidity, so that in a very short period of time it can produce catastrophic damage (\citet{kimmel2004a}). Even if transmitted orally, its influence can be exacerbated by new means of media transmitting and the increasingly need for information. 
\\\indent During World War II, the Americans, under the Office of Strategic Services (OSS), began developing their own rumor-weapon technologies with the help of the scientist Robert Knapp, who also wrote about rumors in an academic context (see \citet{knap1944a}). Knapp's work was adopted by the OSS in 1943 to create a sort of manual for rumor engineers during the war. This document was de-classified in 2004. The manual says that rumors are made memorable by being simple, making concrete references, using stereotypical phrases. In this case,  the suggested model for rumor tries to conclude simple logical concepts to keep its simplicity and phenomenon.
\\\indent Controlling rumors in a country is a soft power (see \citet{nye2004b}) which leads to success in   world politics, economics and etc. In this way, spread rumor is completely related to security of nations. For example, according to CIA declassified documents about  Iran coup 1953, CIA uses spreading  propaganda against  Iran's current  prime minister to undermining his public standing: \textit{This appears to be an example of CIA propaganda aimed at undermining the Prime Minister's public standing, presumably prepared during Summer 1953. It certainly fits the pattern of what Donald Wilber and others after him have described about the nature of the CIA's efforts to plant damaging innuendo in local Iranian media. In this case, the authors extol the virtues of the Iranian character, particularly as admired by the outside world, then decry the descent into "hateful," "rough" and "rude" behavior Iranians have begun to exhibit "ever since the alliance between the dictator  and the Tudeh Party."} (see \citet{ciadoc}).
\\\indent It is important to realize where the starting point for rumor dissemination in the society is (see \citet{kostka2008a}), and how it has been observed. Because different starting point in the network leads to different spreading behavior for a rumor. The structure of society and its homogeneity are two important factors on how rumor statements mutate and how many people can hear and change the propositions of the rumor, and this might lead to an equilibrium point (see \citet{piqueira2010a}). If spread rumor reaches to an equilibrium point, then it clarifies that desires of corresponding agents in the society do not have any  of propositions in the converged rumor. Therefore, rumor spreading in a society can reveal belief indicators of the society. Most mathematical studies of rumor propagation include both theoretical and applied aspects: Graph theory, Dynamic systems and, to a smaller extent, stochastic processes that have played prominent roles in many models of rumor and information propagation (see \citet{olles2006a}, \citet{zhu2010a}, \citet{fu2007a}, \citet{huang2000a}). Daley and Kendal proposed a general model of rumor spreading (see \citet{daley1956a}, \citet{daley2000b}). The Daley-Kendal model and its variants have been used repetitively in the past for quantitative studies of rumor spreading (see \citet{pittel1990a}, \citet{lefevre1944a}, \citet{sudbury1985a}). 
\\\indent Recently, the mathematical models and corresponding simulations on rumor spreading mainly investigated the topology and dynamic of networks, without underlying the logical and psychological aspects of rumor dissemination phenomenon among network's agents (see \citet{zhou2007a}, \citet{olles2006a}, \citet{zanette2001a}). An important shortcoming of the mentioned models is that they either do not investigate change of rumors while spreading or do not take advantage of computational model for describing social interactions and rumor spreading process (see related works in section \ref{relatedworks}).  Therefore, a new model is required, which helps better analysis of rumor changes through agents' interactions.
\\\indent We do not claim that our modeling is better than other counterpart computational models, indeed. Works which consider Markov Decission Process method regard rumor spreading as a macro level of a society, whereas in our work, we study rumor spreading in a micro level regarding agents and the connection network among them. Most of the counterpart models for rumor spreading are less focused on the phenomenon of rumor computationally. This paper proposes consonant mathematical models for both rumor and the social network, which rumors are spreading through; where both rumors and beliefs of agents playing role in rumor dissemination have similar structures. This structure lets researchers to investigate all individuals of the model as vectors containing -1, 0 and 1s. In the some of outstanding counterpart models (such as \citet{piqueira2010a} and \citet{acemoglu2010a}) meeting among agents is modeled based on probable events while transferring information between them and the entire process can be modeled mathematically as MDP where the equilibrium point can be calculated as obtained eigenvalues of the Transition Matrix.
\\\indent There are three Simulation Modeling methods for approximating real life behavior: Discrete Events, System Dynamics and Agent-based Modeling (see \cite{borshchev2004}, \cite{bonabeau2002}). We consider Markov Decission Process as a suitable method to study rumor spreading, but it is categorized in the System Dynamics. Our purpose in this paper is to providing an agent-based model for rumors and analysis of their spreading. It indicates the relationship between the homogeneity of the network and converging of rumor in stable colonies. The research result demonstrates that convergence of spread rumor is independent from spreading process, and network homogeneity is a necessary condition for convergence of the  spread rumor.
\\\indent Simulation of rumor spreading is also vitally important for analyzing and extracting principles from observations. Through simulation we are given the opportunity to evaluate the accuracy of pre assumed principles or obtain initial material for proving related theorem whether the homogeneity of the society is a necessary condition for the convergence of the spread rumor. Beside simulation approach we mathematically have proved theorem \ref{firsttheorem} (see section \ref{sec:mathanal}) about the relation between homogeneity and rumor convergence.
\\\indent In the rest of this paper, in section \ref{computationalmodel}, a computational model for rumor, society and procedure of rumor propagation is described in detail. Section \ref{computersimulation} investigates applied algorithms for simulation of rumor spreading. Section \ref{modelanalysis} analyzes the model and investigates related attributes of a rumor and conditions of rumor convergence. Section \ref{relatedworks} and \ref{futurwork} finish with related works and further work, respectively.

\section{Computational Model}\label{computationalmodel}
Rumor propagation happens as follows:
\\\indent First an actual event happens in the environment, and some agents observe it. The observers based on their desires change the story of the event and send a rumor to their neighbors in the network. After that, each agent, who receives versions of the rumor, produce a new version, based on what he received and also its desire.
\\\indent To describe the environment, we consider a finite set $P=\{p_1,p_2,$ ... $,p_n\}$ of atomic propositions. A \textit{literal} is an atomic proposition or the negation of an atomic proposition. We define a conjunctive literal form (CLF) proposition over P to be a conjunction of literals, where for all atomic propositions $p_i\in P$, either $p_i$ or $\neg p_i$ appears in it, but not both of them. Therefor any CLF proposition $\varphi$ over $P$, is a string $b_1 b_2$ ... $b_n$ in $\{0,1\}^{*}$ with length $n$, where $b_i=1$ if and only if $p_i$ appears in $\varphi$, and $b_i=0$ if and only if $\neg p_i$ appears in $\varphi$.
\\\indent For example, for $P=\{p_1,p_2,p_3\}$,  $\neg p_1\wedge p_2\wedge \neg p_3$ is a CLF proposition, and we represent it by the string $010$. We suppose that, an actual event in the environment or a rumor of it, is a CLF proposition over $P$.
\\\indent A mathematical model for rumor propagation, on a set of atomic propositions $P=\{p_1,p_2,$ ... $,p_n\}$, is a 5-tuple
\begin{equation}
R= \langle Ag,Ob,rumor_P,\delta,\tau \rangle
\end{equation}
where:
\begin{itemize}
  \item $Ag$ is a finite set of agents
  \item $Ob\subsetneqq Ag$ refers to initial observers of the actual event
  \item $rumor_P$ is the set of all CLF formulas over $P$
  \item $\delta$ is a priority function that assigns a real number to each proposition, that is, $\delta : P \rightarrow [0,1]$
  \item $\tau$ is a trust function, that is, $\tau: Ag \times Ag \rightarrow [0,1]$
\end{itemize}

\subsection{Rumor Model}
When an agent witnesses an event, the agent creates a rumor based on its observation of the facts that constitute the event. We refer to the set of agents witnessing the same event as observers, by $Ob$, that are assumed as input channels of rumor spreading system. The rumor is then diffused inside the network. Assuming $P$, to be a set of atomic formulas to describe the environment, an actual event therefore is a CLF formula over $P$. Furthermore, spread rumors of the spreading system are CLF formulas.
\subsection{Agent Model}\label{agentmodel}
It is observed multiple times that agents, who play role in receiving and spreading a story, change its details based on their own desires (see \cite{stern1902a}). In our model, a 4-tuple represents each agent $ag\in Ag$:
\begin{equation}
ag_a=\langle D,Box,\beta,\upsilon \rangle
\end{equation}
where $D$ refers to the desire of the agent, which includes two subsets of $P$, positive subset $\Gamma_a^+$and negative subset $\Gamma_a^-$ so $D=(\Gamma^+,\Gamma^-)$. Interest of the agent determines whether a proposition in $P$ members in $\Gamma_a^+$ or $\Gamma^-$ and $\Gamma^+\cap \Gamma^-=\varnothing$. Any proposition out of the mentioned subsets is supposed to be \textit{inconsiderable} for the agent. We map each agent $ag\in Ag$ to a vector made of $1$, $-1$ and $0$s, which are referred to being in $\Gamma^+, \Gamma^-$ or not being desired, respectively. For example for $P=\{p_1,p_2 ,p_3 ,p_4 ,p_5\}$, agents $ag_0, ag_1$ with desires $D_0=(\Gamma_0^+,\Gamma_0^- )$, $D_1=(\Gamma_1^+,\Gamma_1^- )$, respectively, where $\Gamma_0^+=\{p_1,p_5 \}$, $\Gamma_0^-=\{p_2,p_3 \},\Gamma_1^+=\{p_1,p_3,p_4 \}, \Gamma_1^-=\{p_2 \}$, the corresponding vectors are $\vec{ag_0}  =(1,-1,-1,0,1), \vec{ag_1}  =(1,-1,1,1,0)$. We define $\pi$ as a membership function that is $\pi: Ag\times P\rightarrow\{-1,0,1\}$ and $M$ as considered matrix for $\pi$. So considering matrix for the above example is:
$$M=\left[ \begin{array}{ccccc}
     1 & -1 & -1 & 0 & 1  \\ 1 & -1 & 1 & 1 & 0 
\end{array} \right]$$
\\\indent Each agent has a set indicated by $Box_a\subset rumor_P$, as rumor box, to store all received versions of an event, and the agent combines the received rumors according its desire and then generates a new version to propagate (For more details see the second paragraph of section \ref{sec:dissmodel}). 
\\\indent \textit{Accepting procedure} determines agents' decision of whether to accept or not the resultant of received versions of a rumor. In order to insert mentioned attribute into the model, the function $\beta_a: Rumor\rightarrow \{Yes,No\}$ is defined. The motivation behind defining this function is presented in section \ref{sec:dissmodel}.
\\\indent Let's assume that in the real world, every individual lies with the probability of $1-\upsilon_a$, where $\upsilon_a\in [0,1]$ is the veracity value of agent $ag_a$. Any agent lies (change the value of propositions) to augment the number of propositions that fit its desire. In this model, mentioned value is also used to describe the probability of mutating a rumor before propagating.
\subsection{Network Model}
By rapid growth of social networks, and considering that the major part of the interactions occur among unknown individuals; trust plays an undeniable role in forming the relationship among agents and propagation of information through them. The notion of trust has been used in literature in most researches such as \cite{leskovec2010a} , \cite{guo2011a}, \cite{buskens2002a}. So we consider that this notion is essential in spread rumors. In our modeling, the spreading network is modeled as a complete weighted directed graph $G = (V, E, \tau)$ where each vertex in $V$ represents an agent, and each edge in $E$ represents a connection between a pair of agents. Each edge has a weight value, which refers to trust of the tail agent to the head agent while a rumor get transmitted between them (from head agent as spreader to tail agent as receiver); the function $\tau: E\rightarrow [0,1]$ returns this value. 
\subsubsection{Assumptions}
In order to simplify our model we will state an assumption referring to the network model as follow:
\begin{enumerate}
\item Each agent trusts itself completely:
\begin{equation}\label{mutualtrustass}
\forall a,  \tau(ag_a,ag_a )=1
\end{equation}
\item For a triple of agents ($ag_a ,ag_b,ag_k\in Ag$):
\begin{equation}\label{secondass}
\tau(ag_a ,ag_b )\geq \tau(ag_a ,ag_c )\tau(ag_c ,ag_b )
\end{equation}
This assumption guarantees that the trust value that is raised from direct connection is bigger than relative trust, which is obtained from non-directed connections. If $ag_a$ trusts $ag_c$ with value $\tau(ag_a ,ag_c )$ and $ag_c$ trusts $ag_b$ with value $\tau(ag_c ,ag_b )$, then the trust of $ag_a$ to $ag_b$ is at least $\tau(ag_a ,ag_c )\tau(ag_c ,ag_b )$.
\end{enumerate}
Throughout the rest of this paper we assume that mentioned assumptions are held.
\subsection{Rumor dissemination Model}\label{sec:dissmodel}
In a social system, people constitute a network of connections so that each connection represents how much a couple of agents trust each other mutually. Considering the spreading network $G = (V, E, \tau)$ in this model, each node as a candidate of an agent connects to its neighbor agents. At the beginning, some agents are chosen as \textit{observers}. They observe an actual event and make initial rumors, then spread it through the network. In every generation, an agent is chosen at random from the agent population for taking chance of generating and spreading its generated rumor. A spreader agent calls \textit{hear function} of each neighbor to spread the generated rumor. Each hearer agent $ag$ adds received rumor to its rumor box, $Box$:
\begin{equation}
Box=\{r_1,r_2,...,r_k\}
\end{equation}
where rumor $r_i$ is represented as a binary string:
\begin{equation}
r_i=b_1^i b_2^i...b_n^i \indent b_j^i\in\{0,1\}.
\end{equation}
\\\indent An agent might receive many rumors from different neighbors; the trust value of each connection determines the effectiveness amount of the corresponding received rumor in the merge procedure. When it comes the turn of the agent $ag_a$, first, he starts merging rumors in the $Box_a$. As mentioned, a rumor is a finite string $b_1 b_2 ... b_n$ in $\{0,1\}^*$. Merging procedure gives a new string as the \textit{majority} of preceding strings of received rumors in $Box_a$. In order to obtain the merge procedure's \textit{output} string $r^o = b_1^o b_2^o ... b_n^o$, the following equation is used to calculate the binary value of $j^{th}$ proposition in the merged rumor:
\begin{equation}\label{majorityeq}
b_j^o= \lfloor \frac{1}{2}+\frac{\sum_{r_k\in Box_a}\tau(r_k.spreader ,r_k.reciever)\times b_j^k }{\sum_{r_k\in Box_a}\tau(r_k.spreader ,r_k.reciever)}\rfloor
\end{equation}

The above formula, gives weighted average of bits, which are referred to the binary values of propositions $p_j$ in received rumors, which are stored in the rumor box; so that contribution of each bit $b_j$ is dependent to the trust value between spreader and receiver of the corresponding rumor. Then it checks if the average is bigger than $0.5$, returns 1, otherwise it returns zero. In the rest we define the concept of \textit{accepted propositions} to formalize conflicts between a spread rumor and desire of an agent:
\begin{mydef}\label{acceptedprs}
A proposition $p_i\in P$ in a rumor $r= b_1 b_2...b_n$ is  \textbf{accepted} for agent $ag_a\in Ag$, if the proposition is in the positive subset of the agent's desire and its value is true or it is in the negative subset of the agent's desire and its value is false. In other words when one of the following conditions become satisfied:
\begin{enumerate}
  \item $b_i=1,p_i\in \Gamma_a^+$
  \item $b_i=0,p_i\in \Gamma_a^-$
\end{enumerate}
and it is \textbf{unaccepted} when one of the following conditions become satisfied:
\begin{enumerate}
  \item $b_i=0,p_i\in \Gamma_a^+$
  \item $b_i=1,p_i\in \Gamma_a^-$
\end{enumerate}
otherwise the proposition is supposed to be an \textbf{inconsiderable} proposition. We refer to the set of accepted and unaccepted propositions by $\Pi_a$ and $\Delta_a$, respectively.   
\end{mydef}

After applying the merge procedure to all the propositions included in the received versions in rumor box of $ag_a$, it is time to decide accepting or rejecting the merged version ($r_a^o$). In our simulations, following method is used for accepting function:
\begin{equation}
\beta_a (r_a^o)=\left\{
\begin{array}{rcl}
Yes,&\frac{|\Pi_a^t|}{|P|} > \Theta_a \\No,&otherwise
\end{array}\right.
\end{equation}
\\\noindent where $|\Pi_a^t|$ denotes to the number of accepted propositions (Definition \ref{acceptedprs}) in rumor $r$ for the agent $ag_a$. In the above equation, if ratio of accepted propositions over all propositions goes upper than the accepting threshold $\Theta_a $, the agent decides strictly if he accepts the rumor, otherwise there is not any probability to accept and its turn will be over without spreading any rumor. 
Also through out conveying the rumor, changing some details is possible, therefore the hearer agent might alter some details, based on its desire, which results mutation in the merged version and generate the \textit{final promoted version}. 
\begin{mydef}\label{grumordef}
{The final promoted version of a rumor by agent $ag_a$, in generation $t$ is reffered as \textbf{${\gamma_a}^t$}. This version first became accepted  through the accepting function and then desire of $ag_a$ is applied on it.}
\end{mydef}

In the rest, if the agent wants to alter the value of a proposition of its last promoted version $({\gamma}^{t-1})$ (Definition \ref{grumordef}), then at least one high priority proposition must be selected from unaccepted propositions, using \textit{Roulette wheel selection algorithm} (see \citet{mitchell1998a}). Where the probability of choosing an unaccepted proposition $\tilde{p}_k\in\Delta_a$ is
\begin{equation}\label{rouletteeq}
\mathbb{P} \{choosing \indent\tilde{P}_k \}=\frac{\delta(\tilde{p}_k)}{\sum_{\tilde{p}_j\in \Delta_a}\delta(\tilde{p}_j)}
\end{equation}

The priority of a proposition is the result of its social importance, so that propositions with higher priority preempt those with lower priority when choosing by an agent. There are different ways for defining priorities, which some of them have been suggested in this paper: \citet{eugster2004a}.  After agent $ag_a$ selects target proposition; changes it with the probability of $1-\upsilon_a$, where $\upsilon_a$ denotes to its veracity value. Then the agent spread its generated rumor to the neighbor agents.
\section{Computer Simulation}\label{computersimulation} 
In this section, implemented methods for computer simulation of rumor propagation are investigated. As target of computer simulation, it has to do with the manipulations of symbols using a computer code; more specifically, it uses algorithms and rules to derive mathematical propositions from assumptions (see \citet{mollana2008a}). Therefore written codes and used algorithms will be discussed, which are proposed to help us approaching the target of simulation.
\\\indent As mentioned in the section \ref{sec:dissmodel}, each agent has the chance to spread last generated rumor on its turn. By calling the RUN function (Algorithm 1), first, the agent starts merging received versions in the rumor box using the COBINE-RECEIVED-VERSIONS function (Algorithm 6), second, by calling the CAN-ACCEPT function (Algorithm 5), he checks whether the merged version is acceptable or not and then if the last stage is passed, he changes the rumor based on its desires by calling the MUTATE function (Algorithm 3). Finally the agent spreads the verified version to the neighbors by calling their HEAR functions (Algorithm 4) .

\begin{algorithm}[H]\label{run}
\caption{RUN ( )}
\begin{algorithmic} 
\IF{$CAN-ACCEPT( )$}
\STATE $spreadRumor \leftarrow SPREAD-GENERATED-RUMOR( )$
\FORALL{$e \in EdgesToFriends$}
\STATE HEAR($e.targetAgent , spreadRumor , e.trustCoefficient$)
\ENDFOR
\STATE CLEAR( RUMORBOX )
\STATE PUT( RUMORBOX , $spreadRumor$ ) 
\ENDIF
\end{algorithmic}
\end{algorithm}

\begin{algorithm}[H]\label{sgr}
\caption{SPREAD-GENERATED-RUMOR ( )}
\begin{algorithmic} 
\STATE MUTATE( $acceptedRumor , desire$)
\RETURN $acceptedRumor$
\end{algorithmic}
\end{algorithm}

\begin{algorithm}[H]\label{mutate}
\caption{MUTATE( $acceptedRumor , desire$)}
\begin{algorithmic} 
\FORALL{$p \in rumor.propositions$} 
\IF{IS-POSITIVE-PROPOSITION( $p , desire$ )}
\IF{$p.value$ = FALSE}
\STATE ADD( REJECTED-SET ,$ p$ )
\ENDIF 
\ELSIF{ IS-NEGATIVE-PROPOSITION( $p , desire$ )}
\IF{$p.value$ = TRUE}
\STATE ADD( REJECTED-SET , $p$ ) 
\ENDIF
\ENDIF
\ENDFOR
\STATE $selectedPropsition \leftarrow $ DO-SELECTION( REJECTED-SET )
\IF{RANDOM(0 , 1)$ < 1 - veracityCoefficient$}
\STATE UPDATE CODE($ rumor , selectedPropsition , \neg selectedPropsition.value$ )
\ENDIF
\end{algorithmic}
\end{algorithm}

\begin{algorithm}[H]\label{hear}
\caption{HEAR ( $listenerAgent , rumor , trust$)}
\begin{algorithmic} 
\IF{ $listenerAgent$ .RUMORBOX.CNTAINS(  $rumor$ ) = FALSE}
\STATE PUT($listenerAgent$ .RUMORBOX$ , rumor , trust$)
\ENDIF
\end{algorithmic}
\end{algorithm}

\begin{algorithm}[H]\label{canbel}
\caption{CAN-ACCEPT ( )}
\begin{algorithmic} 
\STATE $acceptedRumor \leftarrow$ COMBINE-RECEIVED-VERSIONS( )
\STATE $nrOfAcceptedPropositions  \leftarrow 0$
\STATE $nrOfRejectedPropositions  \leftarrow 0.01$
\FORALL{$p \in acceptedRumor.propositions$}
\IF{ IS-POSITIVE-PROPOSITION( $p , desire $)}
\IF{ $p.value$ = TRUE}
\STATE $nrOfAcceptedPropositions \leftarrow nrOfAcceptedPropositions+1$
\ELSE
\STATE $nrOfRejectedPropositions \leftarrow nrOfRejectedPropositions+1 $ 
\ENDIF
\ELSIF{IS-NEGATIVE-PROPOSITION( $p , desire$ )}
\IF{ $p.value$ = FALSE }
\STATE $nrOfAcceptedPropositions \leftarrow nrOfAcceptedPropositions+1$
\ELSE
\STATE $nrOfRejectedPropositions \leftarrow nrOfRejectedPropositions+1$
\ENDIF
\ENDIF
\ENDFOR
\STATE $accRatio \leftarrow nrOfAcceptedPropositions /( nrOfAcceptedPropositions + nrOfRejectedPropositions) $
\IF {($acceptedRumor.attractivenessCoefficient+accRatio$) $> acceptingThreshold$}
\RETURN FALSE
\ELSE 
\RETURN TRUE
\ENDIF
\end{algorithmic}
\end{algorithm}

\begin{algorithm}[H]\label{crr}
\caption{COMBINE-RECEIVED-RUMORS ( )}
\begin{algorithmic} 
\STATE $tempRumor \leftarrow network.initialObservation$
\FORALL{$p \in tempRumor.propositions$} 
\STATE $denominator \leftarrow 0.01$
\STATE $enuminator \leftarrow 0$
\FORALL{$r \in RUMORBOX$} 
\IF{ $r.value = TRUE$}
\STATE $val \leftarrow 1$
\ELSE  
\STATE $val \leftarrow 0$
\ENDIF
\STATE $enuminator \leftarrow enuminator + GET(RUMORBOX, r) \times val$
\STATE $denominator \leftarrow denominator + GET(RUMORBOX, r) $
\ENDFOR
\STATE $effRatio\leftarrow enuminator / denominator$
\IF{ $effRatio < 0.5$} 
\STATE UPDATE CODE( $r , p$ , FALSE )
\ELSE 
\STATE UPDATE CODE( $r , p$ , TRUE )
\ENDIF
\ENDFOR
\RETURN $ tempRumor$
\end{algorithmic}
\end{algorithm}

\section{Model Analysis}\label{modelanalysis}
It is expected that if there does not exist any conflict among beliefs, a spread word of mouth through this group could easily converge to consensus version. In our research, we investigate whether homogeneity of society is necessary or/and sufficient condition for the convergence of spread rumor. Therefore in this section we try to answer the following questions: 
\begin{itemize} 
\item If there does not exist conflicting in desires, despite the presence of different versions at the beginning, do they converge to one consensus version after a finite number of generations? (Do all agents accept a similar version of rumor?)
\item In a rumor-society system, if different versions of a rumor reache an equilibrium point (converge), does it ensure that the society was homogeneous?
\end{itemize}
As dissemination of rumors is not independent from members of society (it is affected by desires of agents), so convergence concept has to be clarified so that we can clearly state status of agents in the equilibrium point. In this regard, mathematical concept of \textit{stability} is defined. Also in order to mainstream analysis, we define another prominent concept mathematically; which is \textit{homogeneity}.

\subsection{Stability}
If a system provides its parts with an optimum environment, then they will tend to conserve it. In our rumor spreading system, an agent reaches the optimum point when no change in the details of the received rumor is required. Individual stability extends to social stability by attributing individual stability for all agents.
\subsubsection{Individual Instability}
In the real world, each agent tends to change the content of the received versions of a rumor based on his own desire. As mentioned in the previous chapters, each agent $ag_a$ changes merged version with the probability of $1-\upsilon_a$, the following equation for calculating individual instability value of agent $ag_a$ in $t^{th}$ generation is proposed:
\begin{equation}\label{deltaeq}
\Omega_a^t=|{\Delta_a^*}^t| (1-\upsilon_a )
\end{equation}
where $|{\Delta_a^*}^t|$ denotes to the number of \textit{unaccepted} propositions in ${\gamma_a}^t$ for agent $ag_a$ (see definition 2). The above equation implies that a veridical agent would not apply its desire on the resultant of received versions. In the next section the average value of $\Omega$ on the colony will be used as \textit{momentum individual instability} in each generation of rumor spreading process.
\subsubsection{Social Instability}
Social stability is obtained only if all agents become stable; but predicting an agent's stability is impossible without knowing the status of its neighbor agents. Therefore, the definition of social instability is somehow different from calculating the average of individual instabilities. In fact social instability is a completely time-dependent attribute. We can hereby state that a society is stable if, and only if it will not lose its stability within the next generations, under any circumstances. In the rest, social instability is investigated on sequence of obtained values $I_C (t)$ which denote the mean instability of colony $C$ in $t^{th}$ generation:
\begin{equation}\label{instabilityeq}
I_C (t)=\frac{1}{|C|} \sum_{ag_a\in C}\Omega_a^t 
\end{equation}
\subsection{Homogeneity}
In order to compare colonies with each other, without knowing about their interactions with rumor spreading, we need to define a new network attribute. This attribute is independent from rumor spreading but its impact on rumor spreading must be determined. In this paper, the term of homogeneity refers to the degree to which the desires of agents in society to be compatible. In other words, it is a measure for showing the general lack of tendency to change value of propositions in the spread rumor. The homogeneity attribute for a colony arises from combining the three following measures: 
\begin{enumerate}
  \item The closeness of agents' desires
  \item The corresponding trust value for each connection in the colony
  \item The veracity value of agents
\end{enumerate}

We give a measure for closeness of agents' desires based on their \textit{identical distance}: 
\begin{mydef}\label{identicaldistancedef}
The \textbf{identical distance} between agents $ag_a$,$ag_b\in Ag$, with corresponding vectors $\vec{ag_a}$ and $\vec{ag_b}$ , respectively, is defined:
\begin{equation}
d(\vec{ag_a} ,\vec{ag_b })=\sqrt{\sum_{k=1}^{|P|}\lfloor \frac{|M_{a,k}-M_{b,k}|}{2}\rfloor\times\delta(p_k)^2 }
\end{equation}
where $\delta$ returns the global priority for each proposition $p_k\in P$ and $M$ is the matrix of membership function mentioned in the section \ref{agentmodel}.
\end{mydef}

\begin{mydef}\label{conflictdef}
We say two agents $ag_a,ag_b\in Ag$ conflict, whenever there exists a proposition $p_k\in P$, so that:
$$p_k\in\Gamma_a^+\cap\Gamma_b^- \vee p_k\in\Gamma_a^-\cap\Gamma_b^+$$
\end{mydef}
The following equation is proposed to give a measure for homogeneity of colony based on measure of its heterogeneity:

\begin{equation}\label{homoeq}
h_C=\exp(-\sum_{ag_a\in C}\sum_{ag_b\in C}H_{a,b})
\end{equation}
where $H_{a,b}$ is heterogeneity between agents $ag_a$ and $ag_b$
\begin{equation}\label{hceq}
H_{a,b} = d(\vec{ag_a}  ,\vec{ag_b} ) \tau(ag_a,ag_b )(1-\upsilon_a )
\end{equation}
\\\indent Two things are particularly important in this equation. First, it underlines the idea that the homogeneity of a society depends not only on consistency between agents' desires, but also on the veracity and the trust-based connection between agents. Second, the existence of honest agents (with $\upsilon_a=1$) in a society increases the homogeneity of the society.

\subsection{Mathematical Analysis}\label{sec:mathanal}
In order to analyze the model mathematically, we provide some examples. Each example illustrates how a rumor spread through its corresponding network and which attributes mentioned in the model, effects rumor spreading. Then a theorem on convergence of the spread rumor is proved based on the presented computational model. It is worthwhile to mention that the number of agents in simulations is not significant. The main purpose of given examples is to verify the upcoming theorem. Each example investigates a specific mode of rumor spreading to confirm completeness of considered theorem.
\\\indent In all following examples, we consider the below table that shows a proposition set P that includes 23 propositions and the priority of each proposition is determined subsequently in the second row:
\begin{table}[H]\footnotesize
\centering
\begin{tabular}{|c|c|c|c|c|c|c|c|c|c|c|c|c|c|c|}
\hline
i & 1&2&	3&	4&	5&	6&	7&	8&	9&	10&	11&	12&	13&	14\\\hline
$\delta(p_i)$ &0.8&	0.1&	0.7&	0.3&	0.4&	0.4&	0.6&	0.6&	0.3&	0.2&	0.3&	0.6&	0.2&	0.8\\\hline
\end{tabular}
\begin{tabular}{ |c|c|c|c|c|c|c|c|c|c|c|c|c|c|c|}
\hline
i&15&16&17&18& 19&20&21&22&	23\\\hline
$\delta(p_i)$ &	0.5&	0.5&	0.1&	0.9&	1.0&	0.4&	0.5&	1.0&	0.2\\
\hline
\end{tabular}
\end{table}

We also consider $r_{io}=11101001101110101001010$ as the initial observation of the actual event.

\begin{example}\label{firstex}
There are 9 agents with common veracity value ($\upsilon_i=0.5,i=1,2,... 9$) and $Ob=\{ag_9\}$. For each agent $ag_i\in Ag$, the desires are shown in the Table $2$.
\begin{table}[H]\label{firsttable}
\centering
\begin{tabular}{c |rrrrrrrr}
Sets of desire &\multicolumn{8}{l}{ Included propositions}\\\hline
$\Gamma_i^+$ &1&2&3&7&9&11&13&	17\\
$\Gamma_i^-$&4&	6&8&14&16&20&21&23\\
\end{tabular}
\caption{Desires of agents in Example \ref{firstex}}
\end{table}
In this example the trust function $\tau$ returns 1 for each pair and $h_C=1$.
\begin{figure}[H]\label{firstdiagram}
  \begin{center}
    \includegraphics[width=0.8\textwidth]{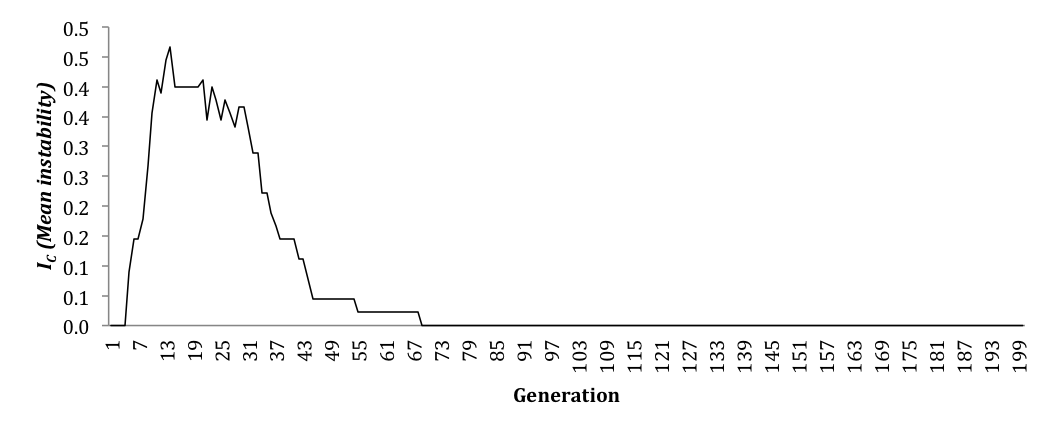}
    \caption{Mean social stability diagram in Example \ref{firstex}.}
\end{center}
\end{figure}
\end{example}
This simulation is designed to investigate the model, with multiple but completely similar agents, which all trust each other. Obviously, the homogeneity value for this society is $1$. Considering Figure \ref{firstdiagram}, the network reached to a stable mode after a number of generations and the spread rumor converged.

\begin{example}\label{secondex}
The same agents in the previous example and the same $Ob$ set, but the trust function $\tau$ returns $0.5$ for each pair of agents. In this example we have $h_C=1$.
\begin{figure}[H]\label{seconddiagram}
  \begin{center}
    \includegraphics[scale = 0.65]{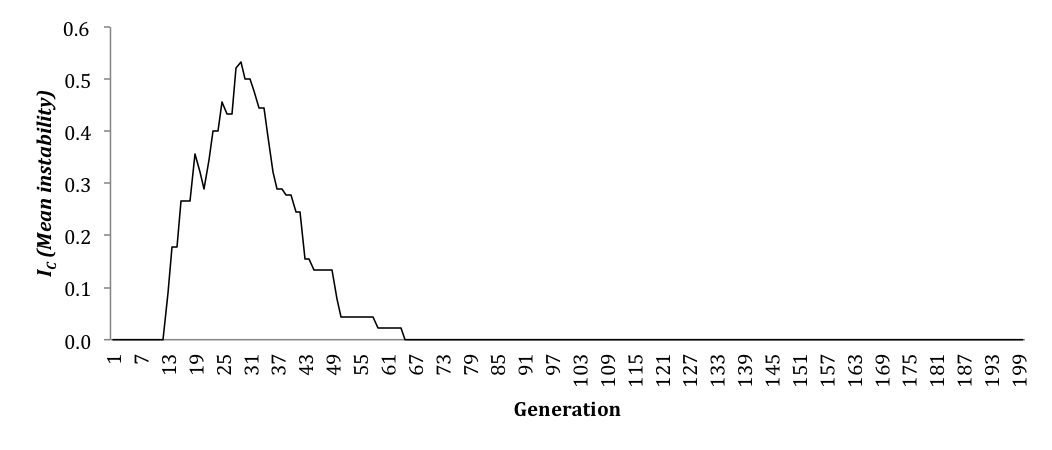}
    \end{center}
  \caption{Mean social stability diagram in Example \ref{secondex}. }
\end{figure}
\end{example}
The only different attribute of this simulation is the trust function. But homogeneity value is not changed apparently. Considering Figure \ref{seconddiagram}, like the previous simulation, the network reached to a stable mode after a number of generations.

\begin{example}\label{thirdex}
There are $2$ agents with $\upsilon_i=0.5$,$i=1,2$ and $Ob=\{ag_1\}$. Considering desires and trust values are shown in the Table $3$ and Table $4$, respectively.
\begin{table}[H]\label{secondtable}
\centering
\begin{tabular}{c |rrrrrrrr}
Sets of desire &\multicolumn{7}{l}{ Included propositions}\\\hline
$\Gamma_1^+$& 1&2&3&10&11&13&17\\
$\Gamma_1^-$&4&6&8&14&20&21&23\\\hline
$\Gamma_2^+$&1&3&7&9&11&13&17\\
$\Gamma_2^-$&4&5&8&14&16&20&22&23\\
\end{tabular}
\caption{Desires of agents in Example \ref{thirdex}}
\end{table}

\begin{table}[H]\label{thirdtable}
\centering
\begin{tabular}{l|c|r}
Agent &1&2\\\hline
1& 1&0.6\\\hline
2& 0.6&1\\
\end{tabular}
\caption{Trust value of connections in Example \ref{thirdex}}
\end{table}
Also in this example we have $h_C=1$.

\begin{figure}[H]\label{thirddiagram}
  \begin{center}
    \includegraphics[scale = 0.65]{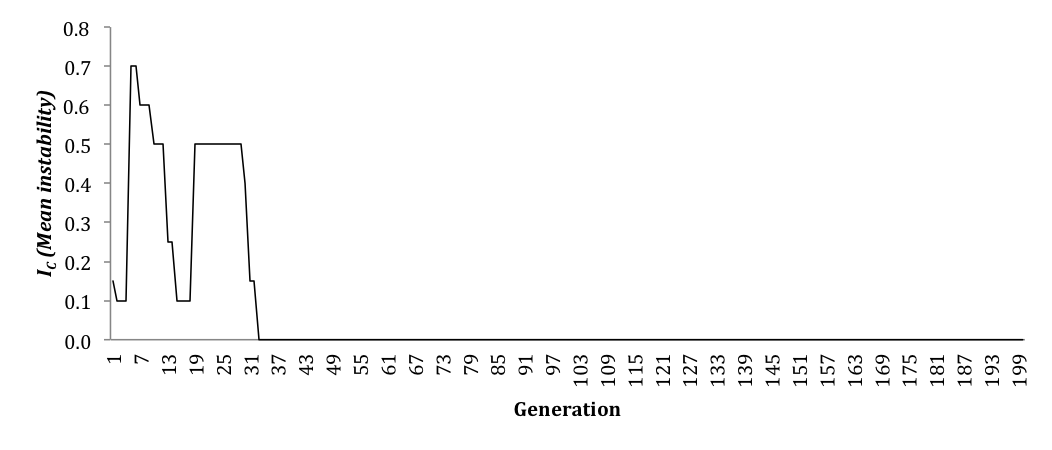}
    \end{center}
  \caption{Mean social stability diagram in Example \ref{thirdex}.}
\end{figure}
\end{example}

\begin{example}\label{fourthex}
There are $2$ agents with $\upsilon_i=0.5,i=1,2$ and $Ob=\{ag_1\}$. Considering desires and trust values are shown in the Table $5$ and Table $6$, respectively.
\begin{table}[H]
\centering
\begin{tabular}{c |cccccccc}
Sets of desire &\multicolumn{7}{l}{ Included propositions}\\\hline
$\Gamma_1^+$&2&5&10&11&13&17&18\\
$\Gamma_1^-$&4&6&8&19&20&23\\\hline
$\Gamma_2^+$&1&3&7&9&12&13&15\\
$\Gamma_2^-$&5&14&16&18&21&22\\
\end{tabular}
\label{fourthtable}
\caption{Desires of agents in Example \ref{fourthex}}
\end{table}

\begin{table}[H]
\centering
\begin{tabular}{c|c|c}
Agent &1&2\\\hline
1& 1&1\\\hline
2& 1&1\\
\end{tabular}
\label{fifthtable}
\caption{Trust value of connections in Example \ref{fourthex}}
\end{table}

In this example we have $h_C=0.3734$.
\begin{figure}[H]\label{fourthdiagram}
  \begin{center}
    \includegraphics[scale = 0.65]{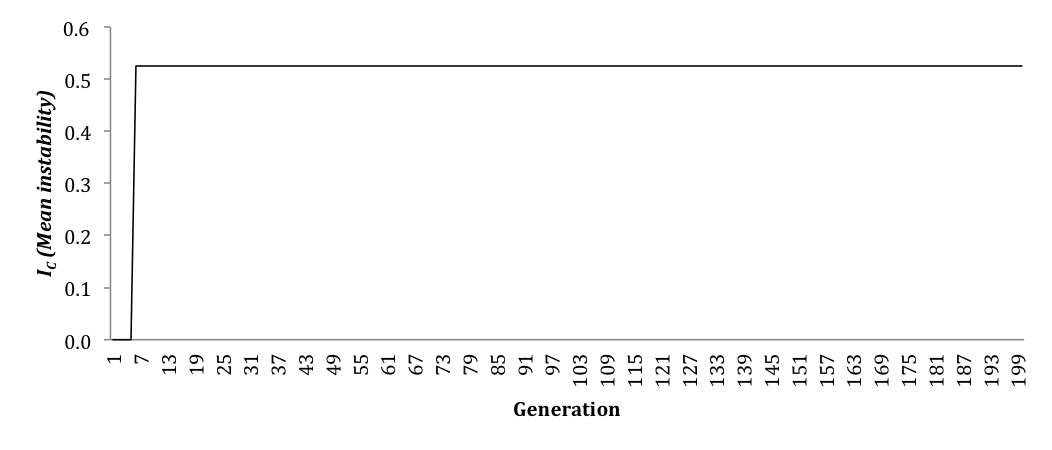}
    \end{center}
  \caption{Mean social stability diagram in Example \ref{fourthex}. }
\end{figure}
\end{example}
Example \ref{thirdex} and \ref{fourthex} are designed to investigate the model supposing $2$ different agents, where in the Example \ref{thirdex}, these $2$ agents make homogeneous network, whereas in the example \ref{fourthex}, agents are inconsistent with each other in their desires, therefore we can make result according to the obtained diagrams (Figure \ref{thirddiagram} and Figure \ref{fourthdiagram}) that the homogeneity value effects the mean stability of the network as the outcome of simulation.

\begin{example}\label{fifthex}
There are $9$ agents with common veracity value ($\upsilon_i=0.5,i=1,2,... 9$) and $Ob=\{ag_9\}$. Considering desires and trust values are shown in the Table $7$ and Table $8$, respectively.

\begin{table}[H]\label{sixthtable}
\centering
\begin{tabular}{c |cccccccc}
Sets of desire &\multicolumn{7}{l}{ Included propositions}\\\hline
$\Gamma_1^+$&1&2&3&9&11&13&17\\
$\Gamma_1^-$&6&14&16&21&23\\\hline
$\Gamma_2^+$&4&12&22\\
$\Gamma_2^-$&5&6&14&21\\\hline
$\Gamma_3^+$&1&3&4&9&13&17\\
$\Gamma_3^-$&6&10&14&18&21\\\hline
$\Gamma_4^+$& 1&2&3&4&9&12&13\\
$\Gamma_4^-$&5&6&14&16&18&23\\\hline
$\Gamma_5^+$&1&2&3&4&12&13&17\\
$\Gamma_5^-$&5&6&14&18&23\\\hline
$\Gamma_6^+$&4&9&12&13&17&22\\
$\Gamma_6^-$&6&14&16&19&21\\\hline
$\Gamma_7^+$&2&4&13&22\\
$\Gamma_7^-$&5&6&16\\\hline
$\Gamma_8^+$&3&12&13&22\\
$\Gamma_8^-$&10&18&23\\\hline
$\Gamma_9^+$&1&2&4&9&12&22\\
$\Gamma_9^-$&5&6&14&19&21
\end{tabular}
\caption{Desires of agents in Example \ref{fifthex}}
\end{table}

\begin{table}[H]\label{seventhtable}
\centering
\begin{tabular}{c|c|c|c|c|c|c|c|c|c}
Agent &1&2&3&4&5&6&7&8&9\\\hline
1&1&0.3&0.3&0.37&0.3&0.32&0.5&0.58&0.33\\\hline
2&0.3&1&0.3&0.35&0.36&0.4&0.52&0.79&0.32\\\hline
3&0.3&0.3&1&0.38&0.36&0.31&0.36&0.58&0.34\\\hline
4&0.37&0.35&0.38&1&0.37&0.32&0.39&0.63&0.34\\\hline
5&0.3&0.36&0.36&0.37&1&0.36&0.32&0.53&0.35\\\hline
6&0.32&0.4&0.31&0.32&0.36&1&0.31&0.39&0.3\\\hline
7&0.5&0.52&0.36&0.39&0.32&0.31&1&0.41&0.33\\\hline
8&0.58&0.79&0.58&0.63&0.53&0.39&0.41&1&0.57\\\hline
9&0.33&0.32&0.34&0.34&0.35&0.3&0.33&0.57&1\\
\end{tabular}
\caption{Trust value of connections in Example \ref{fifthex}}
\end{table}

\begin{figure}[H]
  \begin{center}
    \includegraphics[scale = 0.65]{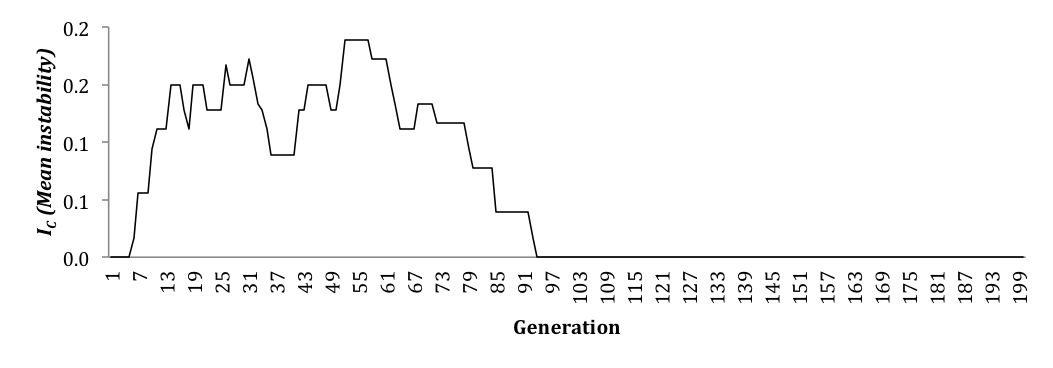}
    \end{center}
  \caption{Mean social stability diagram in Example \ref{fifthex}. The network reached to a stable mode and the spread rumor is converged.}
\end{figure}

\end{example}

\begin{example}\label{sixthex}
There are $9$ agents with common veracity value ($\upsilon_i=0.5,i=1,2,... 9$) and $Ob=\{ag_9\}$. Considering desires are shown in the Table 9 where the trust values are the same as previous example.

\begin{table}[H]\label{eightthtable}
\centering
\begin{tabular}{c |ccccccccc}
Sets of desire &\multicolumn{9}{l}{ Included propositions}\\\hline
$\Gamma_1^+$&1&2&3&7&9&11&13&17\\
$\Gamma_1^-$&4&6&8&14&16&20&21&23\\\hline
$\Gamma_2^+$&4&12&18&22\\
$\Gamma_2^-$&5&6&7&14&20&21\\\hline
$\Gamma_3^+$&1&3&4&7&9&13&17&20&23\\
$\Gamma_3^-$&6&8&10&12&14&18&21\\\hline
$\Gamma_4^+$& 1&2&3&4&7&9&12&13&19\\
$\Gamma_4^-$&5&6&14&18&23\\\hline
$\Gamma_5^+$&1&2&3&4&7&9&12&13&19\\
$\Gamma_5^-$&5&6&14&18&23\\\hline
$\Gamma_6^+$&4&9&12&13&17&22\\
$\Gamma_6^-$&6&11&14&16&19&21\\\hline
$\Gamma_7^+$&2&4&13&19&21&22\\
$\Gamma_7^-$&5&6&11&16&18\\\hline
$\Gamma_8^+$&3&6&12&13&22\\
$\Gamma_8^-$&7&10&18&23\\\hline
$\Gamma_9^+$&1&2&4&9&12&18&22\\
$\Gamma_9^-$&5&6&8&11&14&19&21
\end{tabular}
\caption{Desires of agents in Example \ref{sixthex}}
\end{table}
In this example we have $h_C=1.4\times10^{-7}$.
\begin{figure}[H]
  \begin{center}
    \includegraphics[scale = 0.65]{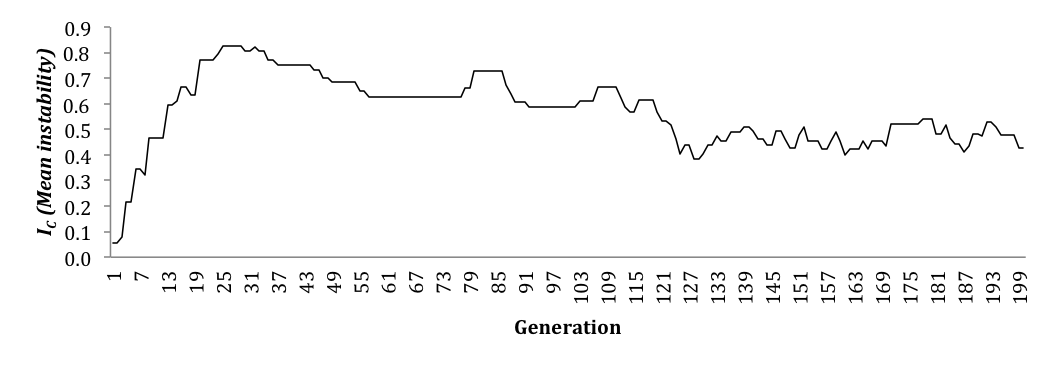}
    \end{center}
  \caption{Mean social stability diagram in Example \ref{sixthex}.}
\end{figure}

\end{example}
The obtained diagrams of simulations in Example \ref{fifthex} and Example \ref{sixthex}, emphasizes on the result, which is already gained from the comparison of Example 3 and Example 4, however there are more agents and connections.

\begin{example}\label{seventhex}
There are $9$ agents with common veracity value ($\upsilon_i=0.5,i=1,2,... 9$) and $Ob=\{ag_9\}$. Considering desires are shown in the Table \ref{ninthtable} where the trust values are the same as previous example.

\begin{table}[H]\label{ninthtable}
\centering
\begin{tabular}{c |cccccccc}
Sets of desire &\multicolumn{8}{c}{ Included propositions}\\\hline
$\Gamma_1^+$&1&2&3&7&9&11&13&17\\
$\Gamma_1^-$&6&8&14&16&20&21&23\\\hline
$\Gamma_2^+$&4&12&22\\
$\Gamma_2^-$&5&6&14&21\\\hline
$\Gamma_3^+$&1&3&4&9&13&17\\
$\Gamma_3^-$&6&8&10&14&18&21\\\hline
$\Gamma_4^+$& 1&2&3&4&9&12&13\\
$\Gamma_4^-$&5&6&14&15&18&23\\\hline
$\Gamma_5^+$&1&2&3&4&12&13&17\\
$\Gamma_5^-$&5&6&7&8&14&18&23\\\hline
$\Gamma_6^+$&4&9&12&13&17&22\\
$\Gamma_6^-$&6&14&16&19&21\\\hline
$\Gamma_7^+$&2&4&13&22\\
$\Gamma_7^-$&5&6&16\\\hline
$\Gamma_8^+$&3&12&13&22\\
$\Gamma_8^-$&10&18&23\\\hline
$\Gamma_9^+$&1&2&4&9&12&18&22\\
$\Gamma_9^-$&5&6&14&19&21
\end{tabular}
\caption{Desires of agents in Example \ref{seventhex}}
\end{table}
In this example we have $h_C=1$.
\begin{figure}[H]
  \begin{center}
    \includegraphics[scale = 0.65]{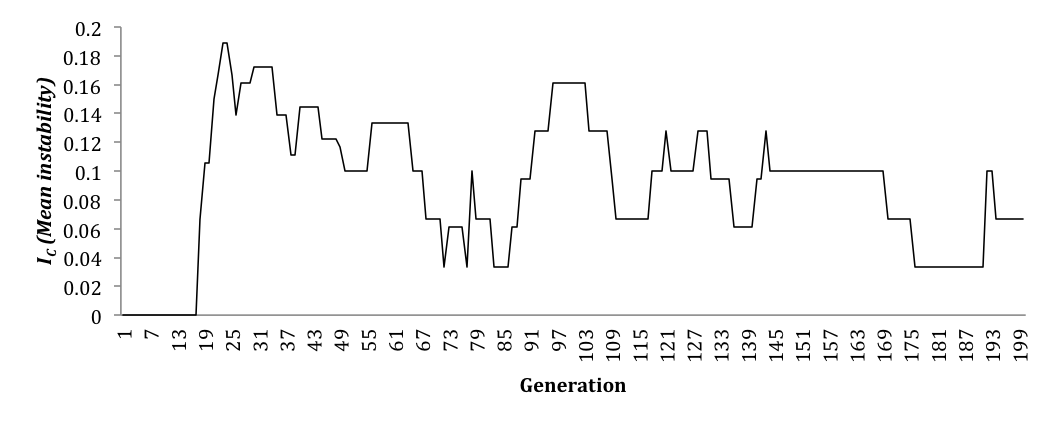}
    \end{center}
  \caption{Mean social stability diagram in Example \ref{seventhex}.}
\end{figure}
\end{example}

The attributes of the simulation in Example \ref{seventhex} is almost similar to Example \ref{sixthex}, but absence of some propositions in desires of a large part of agents, causes appearing unexpected outcome from the simulation process that says the homogeneity is necessary but insufficient condition for stability of colony while spreading a rumor.
\begin{lem}\label{firstlem}
The identical distance between a pair of agents is greater than zero if and only if they conflict with each other:
\begin{equation}
d(ag_i,ag_j)>0 \Leftrightarrow \exists p_k: (p_k\in\Gamma_i^+\cap\Gamma_j^-  \vee  p_k\in\Gamma_i^-\cap\Gamma_j^+)
\end{equation}
\end{lem}

\begin{proof}{
By \ref{identicaldistancedef}:\\
$d(ag_i,ag_j)>0$
$\Leftrightarrow \exists p_k$ : $\lfloor \frac{|M_{i,k}-M_{j,k}|}{2}\rfloor >0$\\
$\Leftrightarrow \exists p_k$ : $(M_{i,k} = 1,M_{j,k}=-1)$ or $(M_{i,k} = -1,M_{j,k}=1)$\\
$\Leftrightarrow \exists p_k$ : $p_k\in\Gamma_i^+\cap\Gamma_j^-  \vee  p_k\in\Gamma_i^-\cap\Gamma_j^+$\\
$\therefore d(ag_i,ag_j)>0 \Leftrightarrow \exists p_k$ : $p_k\in\Gamma_i^+\cap\Gamma_j^-  \vee  p_k\in\Gamma_i^-\cap\Gamma_j^+$.}
\end{proof}
We want to investigate a situation, which two agents $ag_i, ag_j$ conflict in their desires and $ag_i$ acts as sender of a version of a rumor which fit completely to his desires and $ag_j$ plays the role of rumor receiver:   
\begin{lem}\label{secondlem}
If the number of unaccepted propositions for $ag_j$ in generation f (that $ag_j$ spread its final promoted version) was zero ($|{\Delta_j^*}^{f}|=0$) and there exist $ag_i$ which there is conflict between $ag_i$ and $ag_j$ on the value of $p_k$ (see definition 4) and $\tau(ag_i,ag_j )>0$, then if it comes to turn of $ag_i$ in generation $t$, $t>f$ to spread $\gamma_j^t$; the probability of change in the value of $p_k$ in the merged version of $ag_i$ in generation t is greater than zero:
$$(\exists ag_i: p_k\in\Gamma_i^+\cap\Gamma_j^-  \vee  p_k\in\Gamma_i^-\cap\Gamma_j^+ , \tau(ag_i,ag_j )>0)\Rightarrow (\mathbb{P}\{ {b_k^{o_i}}^t+{b_k^{*_i}}^{t-1}=1\}>0)$$
\end{lem}

\begin{proof}
{Assuming $|{\Delta_j^*}^f|=0$ ($f<t$ and ${r^{*_j}}^f$ is stored in $Box_i$), then we consider one of the following conditions for $p_k$:
 \begin{enumerate}
 \item if $p_k\in \Gamma^+_j$ then ${b_k^{*_j}}^f = 1$(i.e. the value of $p_k$ in the generated rumor of $ag_j$ in generation $f$ is true).
 \item if $p_k\in \Gamma^-_j$ then ${b_k^{*_j}}^f = 0$(i.e. the value of $p_k$ in the generated rumor of $ag_j$ in generation $f$ is false).
 \end{enumerate}
 Therefore; if $p_k\in\Gamma_i^+\cap\Gamma_j^-  \vee  p_k\in\Gamma_i^-\cap\Gamma_j^+$, by Eq.\ref{majorityeq} for calculating ${b_k^{o_i}}^t$, the probability of changing the value of $p_k$ in $r^{o_i}$ in generation $t$ is greater than zero. 

 }
\end{proof}

\begin{prethm}\label{firsttheorem}
If spreading a rumor in colony $C$ leads to reach a stable state after a number of generations (i.e. for some $m>0$, for all generations $t >m$, $I_C (t)=0$), then the colony $C$ is homogenous (i.e. $h_C=1$). That is :
$$(\exists m\in N,\forall t>m: I_C (t)=0)\Rightarrow h_C=1$$
\end{prethm}

\begin{proof}
{Let $C$ be a colony that is reached to a stable state after $m$ generations by  Eq.\ref{instabilityeq} we have :$$\forall t>m: I_C (t)=0 $$
$$\Rightarrow \forall t>m:\frac{1}{|C|} \sum_{ag_i\in C}\Omega_i^t  = 0 $$
\\ note that by Eq.\ref{deltaeq}, we have $\Omega_i^t \geq 0; hence $:
$$\forall t>m, \forall ag_i \in C, \Omega_i^t = 0$$
$$\Rightarrow \forall t>m, \forall ag_i \in C:  |{\Delta_i^*}^t| (1-\upsilon_i) = 0$$
We want to prove that if heterogeneity between $ag_i$ and any neighbor agent becomes greater than zero then the probability that agent $ag_i$ becomes instable in generation $t>m$ is greater than zero, and then it can not remain in the stable state. So, we consider one of the following cases for each $ag_i$: \\
{\bf{case 1:}} $ (1-\upsilon_i) = 0$. Therefore; using Eq.\ref{hceq}, for each $ag_j \in C$, we have: $H_{i,j}=0$.\\ 
{\bf{case 2:}} $\forall t>m:  |{\Delta_i^*}^t| = 0$. Therefore; after generation $m$, the probability that any change occurs in the size of ${\Delta_i^*}^t$ (set of unaccpted propositions for $ag_i$ in generation t) is zero. So the probability of change in the value of any proposition $p_k$ in $\gamma_i^t$ while going from generation $t$ to $t+1$ is zero.
$$\forall p_k: \mathbb{P}\{ {b_k^{*_i}}^{t+1}+ {b_k^{*_i}}^t = 1 \} = 0.$$
\\Assuming that heterogeneity between $ag_i$ and any neighbor agent becomes greater than zero, if there exists $ag_j$, $H_{i,j} > 0$, then:
$$ d(\vec{ag_i}  ,\vec{ag_j} ) \tau(ag_i,ag_j )(1-\upsilon_i )>0 .$$
\\ we obtain :   $\exists ag_j:  d(\vec{ag_i}$ , $\vec{ag_j} ) > 0$, $\tau(ag_i,ag_j )>0$
\\ which implies that there is conflict among desires of $ag_i,ag_j$. By Lemma \ref{firstlem} we have: 
\begin{equation}\label{proofeq}
\exists ag_j \exists p_k \in P : (p_k\in\Gamma_i^+\cap\Gamma_j^- \vee p_k\in\Gamma_i^-\cap\Gamma_j^+)
\end{equation}
We assume that it comes to the turn of agent $ag_i$ in generation $T>m+1$  to spread its generated rumor, as mentioned in  section $2.4$, change of last generated rumor is probable while merging and mutating procedure. 
$$\mathbb{P}\{ {b_k^{*_i}}^{t-1}+ {b_k^{*_i}}^t = 1 \} = \mathbb{P}\{ {b_k^{o_i}}^t+{b_k^{*_i}}^{t-1}=1\} \mathbb{P}\{ \beta({r^{o_i}}^t)=Yes\} + \mathbb{P}\{ mutation_{p_k}\}$$  
\\By knowing \ref{proofeq} and using Lemma \ref{secondlem}, we deduce that $\mathbb{P}\{ {b_k^{o_i}}^t+{b_k^{*_i}}^{t-1}=1\} >0$. Also we have:
$$\mathbb{P}\{ \beta({r^{o_i}}^t)=Yes\} = \mathbb{P}\{ \frac{|\Pi_i^t|}{|P|} > \Theta_i\}>0$$
and by Eq.\ref{rouletteeq}
$$\mathbb{P}\{ mutation_{p_k}\}=\frac{\delta(\tilde{p}_k)}{\sum_{\tilde{p}_j\in \Delta_i}\delta(\tilde{p}_j)}\times(1-\upsilon_i) > 0$$
then
$$\mathbb{P}\{ {b_k^{*_i}}^{t-1}+ {b_k^{*_i}}^t = 1 \} > 0.$$
and now we deduce that :
\begin{equation}
\exists ag_j \exists p_k \in P : (p_k\in\Gamma_i^+\cap\Gamma_j^- \vee p_k\in\Gamma_i^-\cap\Gamma_j^+) \Rightarrow \exists p_k :  \mathbb{P}\{ {b_k^{*_i}}^{t-1}+ {b_k^{*_i}}^t = 1 \} > 0
\end{equation}
contradiction. Therefore $H_{i,j}=0$. 
\\In both cases we obtained that $\forall ag_j, H_{i,j}=0$, then $h_C=\exp(-\sum_{ag_i\in C}\sum_{ag_j\in C}H_{i,j}) = 1$.
}
\end{proof}

\section{Related Work}\label{relatedworks}
In this section, we briefly present some of the proposed mathematical models for rumor and information propagation. Most of these models have traditionally used a rumor as epidemic approach that oversimplifies the phenomenon and demographic distribution of the people make role in rumor dissemination. 
\\\indent Kostka et al proposed a simple model in order to examine the diffusion of competing rumors (see \citet{kostka2008a}). They used concepts of game theory and location theory. They modeled the selection of rumors' starting nodes as a strategic game.
\\\indent Wang et al designed a trust model for rumor spreading (\citet{wang2009a}). They considered that, when people exchange information, the trust in the received information is related to the interpersonal closeness. They also used the Monte Carlo method to find the key source nodes in rumor spreading by comparing the total number of spread nodes and spreading time.
\\\indent One of the most challenging models, which are studied in our primary research, was the model proposed by Acemoglu et al (\citet{acemoglu2010a}). They provided a model to investigate the tension between information aggregation and spread of misinformation. In this model, individuals meet pairwise and exchange information by adopting the average of their pre-meeting beliefs and, under some assumptions; convergence of beliefs to a stochastic consensus is obtained. The main results of this paper quantify the extent of misinformation by providing bounds or exact results on the gap between the consensus value and the benchmark value.
\\\indent Nekovee et al (\citet{nekovee2007a}) introduced a general stochastic model for the spread of rumours, and derive mean-field equations that describe the dynamics of the model on complex social networks. They used analytical and numerical solutions of these equations to examine the threshold behavior and dynamics of the model on several models of such networks and showed that in both homogeneous networks and random graphs the model exhibits a critical threshold in the rumor-spreading rate below, which a rumor cannot propagate in the system. The obtained results show that scale-free social networks are prone to the spreading of rumours, just as they are to the spreading of infections. They are relevant to the spreading dynamics of chain emails, viral advertising and large-scale information dissemination algorithms on the Internet.
\\\indent Piqueira (\citet{piqueira2010a}) studied the equilibrium of rumor propagation. In the corresponding paper, in an analogy with the SIR (Susceptible-Infected-Removed) epidemiological model \citet{kermack1933a}, the ISS (Ignorant-Spreader-Stifler) rumor-spreading model is studied, and by using concepts from the Dynamical Systems Theory, stability of equilibrium points is established, according to propagation parameters and initial conditions.

\section{Concluding remarks and further works}\label{futurwork}
In this paper we studied rumor dissemination in societies that all agents can communicate to each other. In our model, rumor spreading does not effect the desires of agents (the set $\Gamma^+$ and $\Gamma^-$ remain unchanged through rumor spreading) and it converges, then the society needs to be homogenous. In other words, if a society is not homogenous, it is not possible that spread rumor converges. 
\\\indent We are going to improve our model for different goals in our further works. We may consider societies that the underlying graph is not a complete one. In this case, we need to cluster the society and find that in which subsection a rumor converges.
\\\indent We restricted ourselves to conjunctive literal forms (CLF) formulas to describe events and rumors. Conjunctive normal forms  (CNF) formulas are more expressive and as a further work, we may use CNF formulas to formally model events and rumor.


\newpage
\renewcommand{\refname}{References}
\bibliography{references} 
\end{document}